\documentclass[prb,aps,floatfix,amsmath,amssymb,twocolumn,superscriptaddress,tightenlines]{revtex4}
\usepackage{graphicx}
\usepackage{epstopdf}
\usepackage[all]{xy}
\newcommand{\be}{\begin{equation}}
\newcommand{\ee}{\end{equation}}
\usepackage{amsmath}
\usepackage{amsfonts}
\usepackage{bm}
\usepackage{color}
\usepackage{subfigure}
\usepackage{amsthm}
\def\Width{0\kern2\tabcolsep\ldots\kern1\tabcolsep0}

\newtheorem{lem}{Lemma}
\newtheorem{thm}{Theorem}
\newtheorem{defi}{Definition}
\newtheorem{corollary}{Corollary}

\newcommand{\drawgenerator}[8]{%
\xymatrix@!0{%
& #8 \ar@{-}[ld]\ar@{.}[dd] \ar@{-}[rr] & & #7 \ar@{-}[ld]  \\%
#1 \ar@{-}[rr] \ar@{-}[dd] &  & #2 \ar@{-}[dd] &            \\%
& #6 \ar@{.}[ld] &  & #5 \ar@{-}[uu] \ar@{.}[ll]       \\%
#3 \ar@{-}[rr] &  & #4 \ar@{-}[ru]                       %
}%
}

\begin{document}
\title{3D local qupit quantum code without string logical operator}
\author{Isaac H. Kim}
\affiliation{Institute of Quantum Information, California Institute of Technology, Pasadena CA 91125, USA}

\date{\today}

\begin{abstract}
Recently Haah introduced a new quantum error correcting code embedded on a cubic lattice. One of the defining properties of this code is the absence of string logical operator. We present new codes with similar properties by relaxing the condition on the local particle dimension. The resulting code is well-defined when the local Hilbert space dimension is prime. These codes can be divided into two different classes: the local stabilizer generators are either symmetric or antisymmetric with respect to the inversion operation. These is a nontrivial correspondence between these two classes. For any symmetric code without string logical operator, there exists a complementary antisymmetric code with the same property and vice versa. We derive a sufficient condition for the absence of string logical operator in terms of the algebraic constraints on the defining parameters of the code. Minimal number of local particle dimension which satisfies the condition is $5$. These codes have logarithmic energy barrier for any logical error. 
\end{abstract}
\maketitle

\section{Introduction}
Quantum error correcting code is a powerful tool for protecting quantum information from undesirable external noises. Particularly interesting class of quantum error correcting codes are the ones with local generators. Such codes are usually not a `good' quantum code in a conventional sense, but they possess a special property of being naturally resilient against local noises. These codes can store quantum information in the topological degrees of freedom. From the point of view of quantum many body systems, the local generators can be thought as the local terms in the quantum many-body hamiltonian. Low-energy excitation spectrum consists of particles, or more generally extended objects.\cite{Kitaev2003,Bombin2006,Bombin2007,Yoshida2011,Nussinov2008,Nussinov2009,Kim2011} Presence of particles poses a potential problem for protecting quantum information from thermal noise. When particles and antiparticles are created out of vacuum, they can diffuse freely and cause a logical error. In a 4D system, it is possible to construct a model in such a way that low energy spectrum consists of closed strings with finite string tension.\cite{Alicki2008} The thermal excitations are suppressed for low enough temperature, allowing a self-correction below a critical temperature.

There had been some loose gauge theoretic evidences suggesting that an analogous self-correction phenomenon may be difficult to achieve in 3D system, until Haah recently introduced a code that can potentially prove otherwise.\cite{Haah2011} The code allows the existence of point-like defects, but their diffusion is constrained by a logarithmic energy barrier. Its ground state structure is unconventional in a sense that the number of encoded qubits change as the system size change, but this is what we expect due to Yoshida's result.\cite{Yoshida2011} Otherwise such systems can only have a constant energy barrier to disrupt the encoded quantum information. Haah and Bravyi showed that there is a logarithmic energy lower bound due to so called `no string rule.'\cite{Bravyi2011}

No string rule can be heuristically understood in a following way. Suppose we are given a set of defects which are sufficiently localized in a finite ball with radius $r$. What is the maximum distance $l$ that these defects can travel without creating more excitations? For toric code and its variants, there are defects that can travel indefinitely far without creating more excitations. Space-time trajectory of the defects are stringlike, and the excitations are created only at the open ends of the string. As the open string extends, the defects can travel accordingly. When they wind around the torus, this might lead to a logical error. Haah's code explicitly forbids such possibility. It is possible to prove that for one of those codes, $l \leq O(r)$. The upper bound on the ratio between the size of the defect and its maximum diffusion distance was introduced as `aspect ratio' by Bravyi and Haah. Given a constant aspect ratio, it is possible to prove a logarithmic energy barrier for the energy diffusion cost of the defects.\cite{Bravyi2011} Energy lower bound for the logical error follows trivially from this result.

These results strongly suggest a viability of constructing self-correcting quantum memory in 3D from finite strength short range interaction. Indeed, the decoherence time scales polynomially system size as long as the system size is not too large.\cite{Bravyi2011a} An interesting question to ask is if such systems exist in nature, or if there are other theoretical models with same properties. We claim such physical systems exist in a more general context, at least theoretically. We relax the constraint on the local particle dimension and study if there are systems with similar properties.

Arguments for establishing logarithmic energy barrier can be generalized to any stabilizer code as long as one can prove no string rule. This applies to a stabilizer code for qudits as well, but the stabilizer formalism is valid only when the dimension of the particle is a power of prime number. Haah's code can be thought as having a local particle dimension $d=2^2$. Here we consider the other scenario, where the local particle dimension is a prime number.

Assuming the prime dimension, translational invariance, and unique stabilizer generator for each unit lattice, the generators defining the code can be classified into two different families. They are either symmetric or antisymmetric under the inversion operation. These two families are related to each other in a subtle way. Bulk property of the codes are identical up to local unitary transformation. This relation continues to hold under periodic boundary condition, but only at particular length scale. Fortunately, the presence or absence of string logical operator comes from the bulk property of the code. Hence one can deduce a one-to-one correspondence between a family of symmetric and antisymmetric code. This will be explained in greater detail in the paper.

To compare our code with Haah's code, Haah's code has two local stabilizer generators for each cube, which corresponds to the generators responsible for the bit flip and phase flip error. Our code has one stabilizer generators for each cube. Some of Haah's code is a CSS code, while none of our code is. Perhaps more importantly, local particle dimension of Haah's code is $2^2$, while in our code it is a prime number $p$. We shall in fact see that $p=2$ inevitably leads to an existence of string logical operator, which confirms the numerical result. As we shall see throughout the rest of the paper, the main difference comes from the structure of the base field: base field for our code and Haah's code is $\text{GF}[p]$ and $\text{GF}[2^2]$.

Our main contribution is a discovery of simple sufficient condition for checking the absence of string logical operator. Given a stabilizer code with cubic local generators, no string rule is implied by simple algebraic constraints on the parameters of the code over a finite field $\mathbb{F}_p$. Existence of quantum code without string logical operator for $p\geq5$  follows from this result. It is important to note that, however, there is a strong numerical evidence that $p=3$ code may as well be free of string logical operator. When $p=2$, one can prove that string logical operator always exists, confirming Haah's result.

We begin by briefly reviewing stabilizer formalism for qudit in Section \ref{sec:stabilizer}. In Section \ref{sec:constraints}, we introduce a set of reasonable constraints that we impose on the code to prove the absence of string logical operator. Out of the codes which are absent of string logical operator, some of them can be mapped into each other via local unitary transformation or lattice symmetry operation. We establish such equivalence relation in Section \ref{sec:equivalence}. Logical operators and their commutation relations are discussed in Section \ref{sec:logicaloperator}. We conclude with discussion and open problems in Section \ref{sec:conclusion}.

\section{Stabilizer codes\label{sec:stabilizer}}
Our code is a special class of quantum error correcting code which are known as stabilizer codes. Most of the well-known codes in this categories are binary codes, but a generalization to any $d$-dimensional Hilbert space is straightforward as long as $d$ is prime.\cite{Ketkar2006} Our code is not of CSS-type. In other words, our code cannot be manifestly divided into two classical codes each of which are exclusively responsible for either the bit flip error or phase flip error. For a $d$-dimensional Hilbert space, there exists a complete set of orthonormal unitary matrices. These are generalization of Pauli matrices in $d=2$. These can be expressed as a product of generalized $X$ and generalized $Z$ operators. Generalized $X$ operator can be thought as a discrete momentum operator with a periodic boundary condition, and $Z$ operator can be thought as a discrete position operator in the same setting. In standard basis, they can be expressed as the following.
\begin{align}
X_{ij}&=1 \quad j=i+1 \mod d \\
&= 0 \quad \text{otherwise}.\\
Z_{ij}&= \omega^{i} \quad j=i \\
&=0 \quad \text{otherwise},
\end{align}
where $\omega$ is the $d$th root of unity. We shall denote $X^{\alpha_1}Z^{\alpha_2}$ as $S_{\alpha}$, with symplectic pair $\alpha= (\alpha_1, \alpha_2)$. As in the binary code case, $X$ and $Z$ satisfy nontrivial commutation relation.
\begin{equation}
XZ=ZX\omega.
\end{equation}

In this setting, stabilizer generator is described in FIG.\ref{fig:generator}. The stabilizer group is generated by the translation of these generators in three different directions. We expect these generators to commute with each other, but the set of solutions from such constraint is too large. We will defer this procedure to the latter part of the paper. The reason becomes much more transparent when we apply the commuting condition in conjunction with the no string rule condition. The solutions satisfying such constraints altogether becomes much easier to classify.

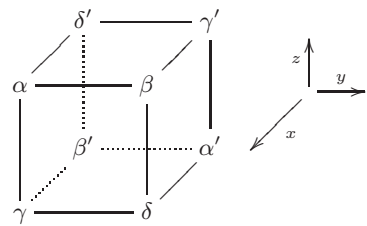
\begin{figure}[h]
\centering
\begin{tabular}{cc}
$ \drawgenerator{\alpha}{\beta}{\gamma}{\delta}{\alpha'}{\beta'}{\gamma'}{\delta'} $

\xymatrix@!0{%
   &      &      \\%
& \ar[ld]^x \ar[u]^z \ar[r]^y & \\%
& & %
}
\end{tabular}
\caption{Stabilizer generator before enforcing any assumption.}
\label{fig:generator}
\end{figure}

Throughout the analysis, we will frequently encounter a \emph{symplectic product}.
\begin{equation}
\langle \alpha,  \beta \rangle = \alpha_1\beta_2 - \alpha_2 \beta_1.
\end{equation}
It has following two useful properties, which we shall use extensively throughout the paper.
\begin{equation}
S_{\alpha} S_{\beta} = S_{\beta} S_{\alpha} \omega^{\langle\alpha , \beta\rangle},
\end{equation}
where
\begin{equation}
S_{\alpha} = X^{\alpha_1}Z^{\alpha_2}
\end{equation}

\begin{lem}
For $\alpha\neq (0,0)$, and $d$ a prime number,
\begin{equation}
\langle \alpha , \beta \rangle=0
\end{equation}
if and only if $\beta=a\alpha$ for some $a\in \mathbb{Z}_d$. \label{lemma:zerosymplecticpair}
\end{lem}
These can be all checked with simple algebra. We will leave the proof as an exercise. We will assume all the arithmetic operation is performed on a finite field $\mathbb{Z}_d$ unless specified otherwise.

\subsection{Measurement of syndromes}
Pauli matrices are both unitary and hermitian for quantum binary codes, but the same statement does not hold for higher dimensional systems. We would like to discuss how the measurements can be done in our model. When $d$ is an odd prime number,
\begin{equation}
P_{\alpha}(r)=\frac{1}{d}\sum_{m=0}^{d-1}(\omega^rS_{\alpha})^m
\end{equation}
is a complete set of orthogonal projections.\cite{Pittenger2004} When applied to our code, the analogous syndrome measurement can be described as
\begin{equation}
P_{s,r} = \frac{1}{d} \sum_{m=0}^{d-1}s^m\omega^{rm},
\end{equation}
where $s$ is the stabilizer generator. As in the standard stabilizer code, it does not matter which value of $r$ we choose for error correction as long as the same convention is used throughout the whole procedure.

As we shall see in latter section, the stabilizer generators turn out to be either symmetric or antisymmetric under the inversion operation. The syndrome measurement operator is invariant for symmetric case, while it changes nontrivially for antisymmetric case. More precisely,
\begin{align}
\mathcal{P} P_{s,r} \mathcal{P} = P_{s,r}
\end{align}
for symmetric code, while
\begin{align}
\mathcal{P} P_{s,r} \mathcal{P} = P_{s,-r}
\end{align}
for antisymmetric code.

One interesting aspect of this formulation is the possibility of finding a quantum many-body system whose ground state can be described by a sum of local bounded-norm hamiltonian. The hamiltonian can be written as
\begin{equation}
H= -\sum_i P_{s_i,0},
\end{equation}
where $i$ is the location of the center of the cubes. The hope is that the energy barrier for logical error is large enough so that the system can protect quantum information from thermal noise.

\section{Constraints \label{sec:constraints}}
The first condition we impose on the generators is the commutation relation, but the number of solutions from this condition alone is too large. Condition we impose on the system is the `deformability condition.' This corresponds to the first step in Haah's method for showing the no string rule. The rough idea is as follows. We want to prove the absence of string logical operator by contradiction. Suppose there is such logical operator. In general, they can be supported on a cylinder with finite cross section. Depending on the choice of stabilizer generators, it is possible to deform this logical operator to an equivalent logical operator which is `flattened' in a sense. This procedure is depicted in FIG.\ref{fig:deformation_procedure}

We shall show this deformability condition, combined with the commuting condition implies the stabilizer group to be either symmetric or antisymmetric under inversion.  This reduces the number of undecided symplectic pairs from $8$ to $4$. The resulting stabilizer generator is FIG.\ref{fig:generator_inversion}. Throughout the analysis, we shall use the terminologies coined by Haah. These are listed in Appendix.

\subsection{Deformability}
Suppose there exists a string logical operator whose support is confined on a cross section with height $h$ and width $w$. Certain models allow such string operators to be deformed onto a flat surface. We will first explain the procedure and then see what kind of condition is necessary for such procedure to be possible. Suppose the logical operator is supported on a $h\times w \times l$ cylinder, where $l$ is the length in the direction perpendicular to the cross section. Pick one of the sites on the edge of the cube. Two stabilizer generators in the cylinder share nontrivial support with this site. Multiply the logical operator with a combination of these stabilizer generators so that the resulting operator acts trivially on that site. This procedure is possible if
\begin{equation}
\det
\begin{pmatrix}
\alpha_1 & -\alpha_2\\
\beta_1 & -\beta_2
\end{pmatrix}
\neq 0,
\end{equation}
or perhaps in simpler term, $\langle \alpha , \beta \rangle \neq 0$, where $\alpha,\beta$ are two symplectic pairs that share nontrivial support with the site. Applying the same logic to other directions, we conclude that any two symplectic pairs $\alpha,\beta$ lying on a same edge must satisfy $
\langle \alpha , \beta \rangle \neq 0$.
\begin{figure}[h]
\centering
\begin{tabular}{c}
\xymatrix{
&  \ar@{-}[ld]\ar@{.}[dd] \ar@{-}[rr] & & \ar@{-}[ld]  \ar@{-}[rr]& &\ar@{-}[ld]\\
 \alpha\ar@{-}[rr] \ar@{-}[dd] &  &  \beta\ar@{-}[dd]\ar@{-}[rr] & & \gamma          \\
&  \ar@{.}[ld] &  &  \ar@{.}[uu] \ar@{.}[ll]  \ar@{.}[rr] & &  \ar@{-}[uu]  \ar@{-}[ld] \\
 \ar@{-}[rr] &  &  \ar@{.}[ru]    \ar@{-}[rr] & &\ar@{-}[uu]
 } \\
\xymatrix{
&  \ar@{-}[ld]\ar@{.}[dd] \ar@{-}[rr] & & \ar@{-}[ld]  \ar@{-}[rr]& &\ar@{-}[ld]\\
 \alpha'\ar@{-}[rr] \ar@{-}[dd] &  &  0 \ar@{-}[dd]\ar@{-}[rr] & & \gamma'          \\
&  \ar@{.}[ld] &  &  \ar@{.}[uu] \ar@{.}[ll]  \ar@{.}[rr] & &  \ar@{-}[uu]  \ar@{-}[ld] \\
 \ar@{-}[rr] &  &  \ar@{.}[ru]    \ar@{-}[rr] & &\ar@{-}[uu]
 } \\
\xymatrix{
&  \ar@{-}[ld]\ar@{.}[dd] \ar@{-}[rr] & & \ar@{-}[ld]  \ar@{-}[rr]& &\ar@{-}[ld]\\
 \alpha'\ar@{-}[rr] \ar@{-}[dd] &  &  0 \ar@{-}[dd]\ar@{-}[rr] & & 0          \\
&  \ar@{.}[ld] &  &  \ar@{.}[uu] \ar@{.}[ll]  \ar@{.}[rr] & &  \ar@{-}[uu]  \ar@{-}[ld] \\
 \ar@{-}[rr] &  &  \ar@{.}[ru]    \ar@{-}[rr] & &\ar@{-}[uu]
 }
 \end{tabular}
 \caption{This diagram represents a deformation procedure. One can multiply a suitable choice of
 stabilizer elements so that the action of string operator on $B$ is $0=(0,0)$. If two symplectic
 pairs on the diagonal line are linearly dependent to each other, one can further deform $C'$ into $0$. Repeat this procedure until we get rid of the entire line.\label{fig:deformation_procedure}}
\end{figure}
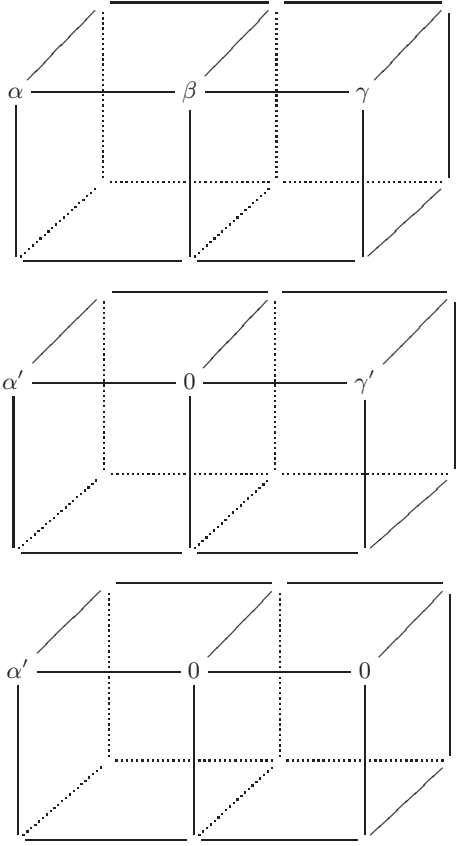

Since we started from a logical operator and multiplied it by the stabilizer group element, this operator must still commute with all the stabilizer generators. In particular, note that there are stabilizer generators that share nontrivial support with the string logical operator at single site. Using Lemma \ref{lemma:zerosymplecticpair}, one can see that the only way to get rid of the nontrivial support on these sites is to force the relation $\alpha = a \alpha'$ for some $a \in \mathbb{Z}_d$. Applying the same logic in all three directions, we obtain FIG.\ref{fig:generator_deformability}.
\begin{figure}[h]
\centering
\begin{tabular}{cc}
$ \drawgenerator{\alpha}{\beta}{\gamma}{\delta}{a\alpha}{b\beta}{c\gamma}{d\delta} $

\xymatrix@!0{%
   &      &      \\%
& \ar[ld]^x \ar[u]^z \ar[r]^y & \\%
& & %
}
\end{tabular}
\caption{Stabilizer generator assuming commutation relation.}
\label{fig:generator_deformability}
\end{figure}
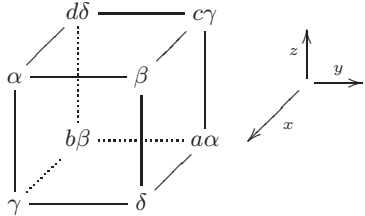

The next step is to use the commutation relation to further narrow down the constraints between the numbers $a,b,c,d$. Note that these generators manifestly commute with each other whenever they meet a single site. By forcing them to commute with each other whenever they meet at two sites, we can obtain $a=b=c=d$ and $a^2=1$. Since this equation over $\mathbb{Z}_p$ for some prime number $p$, the solution is $a=\pm 1$, which corresponds to FIG.\ref{fig:generator_inversion}. We shall denote the symmetry code as $\mathcal{C}^{\alpha \beta \gamma \delta}_{S}$ and $\mathcal{C}^{\alpha \beta \gamma \delta}_{A}$.
\begin{figure}[h]
\centering
\begin{tabular}{cc}
$ \drawgenerator{\alpha}{\beta}{\gamma}{\delta}{\alpha}{\beta}{\gamma}{\delta} $
$ \drawgenerator{\alpha}{\beta}{\gamma}{\delta}{\bar{\alpha}}{\bar{\beta}}{\bar{\gamma}}{\bar{\delta}} $
\\
\xymatrix@!0{%
   &      &      \\%
& \ar[ld]^x \ar[u]^z \ar[r]^y & \\%
& & %
}
\end{tabular}
\caption{Stabilizer group generators for $\mathcal{C}^{\alpha \beta \gamma \delta}_{S}$ and $\mathcal{C}^{\alpha \beta \gamma \delta}_{A}$. $\bar{\alpha} = -\alpha$.}
\label{fig:generator_inversion}
\end{figure}
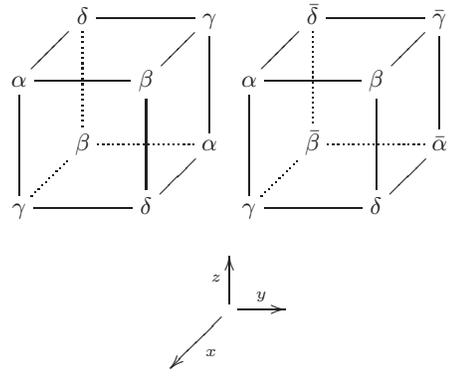

Combining the commutation relation and deformability condition, we arrive at the following conclusion. For code $\mathcal{C}^{\alpha \beta \gamma \delta}_{S,A}$, if $\langle A, B\rangle \neq 0$ $\forall A\neq B$, where $A,B \in \{\alpha, \beta, \gamma, \delta \}$, logical string segment can be mapped into an equivalent flat segment.

Immediate consequence of this configuration is that if $d=2$, it is impossible to come up with $\alpha,\beta, \gamma, \delta$ satisfying the conditions introduced so far. Since there are only $2^2-1=3$ nontrivial symplectic pairs, there has to be at least two of these four pairs which are identical, or one of them must be $(0,0)$. Either case, there always exists at least a pair of symplectic pairs $\alpha, \beta$ such that $\langle \alpha , \beta \rangle= 0$. Hence the minimal local particle dimension that can satisfy these conditions is $d=3$.

\begin{figure}[h]
\centering
\begin{tabular}{c}
\setlength{\unitlength}{0.5cm}
\begin{picture}(20,20)

\linethickness{0.075mm}
\multiput(0,6)(1,0){4}%
{\line(0,1){3}}
\multiput(0,6)(0,1){4}%
{\line(1,0){3}}

\put(4,7.5){\vector(1,0){1}}

\multiput(6,6)(1,0){3}%
{\line(0,1){3}}
\multiput(6,6)(0,1){4}%
{\line(1,0){2}}
\multiput(9,6)(1,0){1}%
{\line(0,1){2}}
\multiput(8,6)(0,1){3}%
{\line(1,0){1}}

\put(10,7.5){\vector(1,0){1}}

\multiput(12,6)(1,0){2}%
{\line(0,1){3}}
\multiput(12,6)(0,1){4}%
{\line(1,0){1}}
\multiput(14,6)(1,0){1}%
{\line(0,1){2}}
\multiput(13,6)(0,1){3}%
{\line(1,0){1}}
\multiput(15,6)(1,0){1}%
{\line(0,1){1}}
\multiput(14,6)(0,1){2}%
{\line(1,0){1}}

\put(13.5,5){\vector(0,-1){1}}

\multiput(12,0)(1,0){1}%
{\line(0,1){3}}
\multiput(12,0)(0,1){3}%
{\line(1,0){1}}
\multiput(13,0)(1,0){1}%
{\line(0,1){2}}
\multiput(13,0)(0,1){2}%
{\line(1,0){1}}
\multiput(14,0)(1,0){1}%
{\line(0,1){1}}
\multiput(14,0)(0,1){1}%
{\line(1,0){1}}

\put(5,1.5){\vector(-1,0){1}}

\multiput(6,0)(1,0){1}%
{\line(0,1){3}}
\multiput(6,0)(0,1){2}%
{\line(1,0){1}}
\multiput(7,0)(1,0){1}%
{\line(0,1){1}}
\multiput(7,0)(0,1){1}%
{\line(1,0){2}}

\put(11,1.5){\vector(-1,0){1}}

\multiput(0,0)(1,0){1}%
{\line(0,1){3}}
\multiput(0,0)(0,1){1}%
{\line(1,0){3}}
\end{picture}
\end{tabular}
\caption{Deformation of string logical operator when viewed from the direction normal to its cross section. \label{fig:deformation}}
\end{figure}
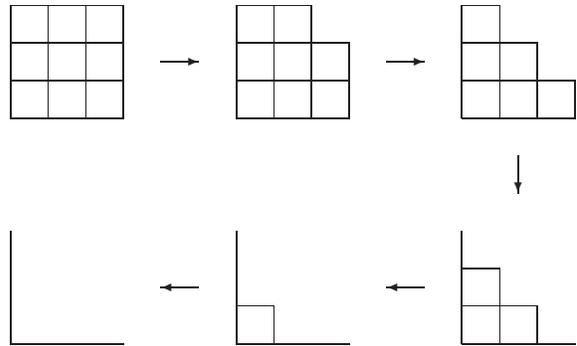

For these codes, any string logical operator with finite thickness can be deformed into an operator having nontrivial support only on a surface. In general, if we started with a string operator with cross section width $w$ and height $h$, such logical operator can be confined to a surface with a kink. See FIG.\ref{fig:deformation}.

\subsection{No string rule}
So far we reduced a number of codes whose string logical operator, if there exists any, can be flattened to a surface. We would like to narrow down a number of codes further to the ones which we can prove the absence of such logical operators. In Haah's code, it is possible to prove statement by simply computing the constraints on the boundary, but the same method does not work for our code. Without loss of generality, consider two consecutive sites on the boundary. There are two stabilizer generators that share nontrivial support with \emph{both} of these sites. These two give a constraint, and there are $4$ unknown variables, which corresponds to the $2$ symplectic pairs. One can easily check that these linear equations that relate the unknowns are linearly independent. Therefore, once one of the symplectic pair is specified, other one is completely determined.

For example, if the pair of symplectic pairs that meet on two consecutive sites are $(\alpha, \beta)$ and $(\gamma, \delta)$, the resulting equation is the following.
\begin{align}
\langle \alpha,  A \rangle + \langle \beta , B\rangle &= 0 \\
\langle \gamma , A \rangle + \langle \delta , B \rangle &=0,
\end{align}
or in matrix form,
\begin{equation}
\begin{pmatrix}
\alpha_1 & -\alpha_2 \\
\gamma_1 & - \gamma_2
\end{pmatrix}
\begin{pmatrix}
A_1 \\
A_2
\end{pmatrix}
+
\begin{pmatrix}
\beta_1 & -\beta_2 \\
\delta_1 & - \delta_2
\end{pmatrix}
\begin{pmatrix}
B_1 \\
B_2
\end{pmatrix}
=0
\end{equation}

Here $A$ and $B$ are the unknown symplectic pairs on two consecutive sites. See FIG.\ref{fig:boundary_constraint}.

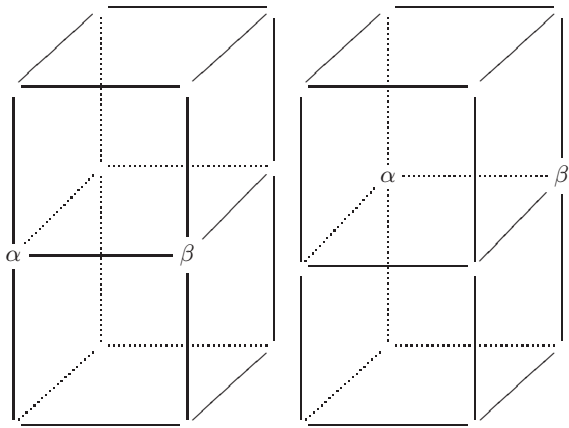
\begin{figure}[h]
\centering
\begin{tabular}{cc}
\xymatrix{
&  \ar@{-}[ld]\ar@{.}[dd] \ar@{-}[rr] & & \ar@{-}[ld]  \\
 \ar@{-}[rr] \ar@{-}[dd] &  &  \ar@{-}[dd] &            \\
&  \ar@{.}[ld] &  &  \ar@{-}[uu] \ar@{.}[ll]       \\
\alpha \ar@{-}[rr] &  & \beta \ar@{-}[ru]       \\
&  \ar@{.}[uu]\ar@{.}[ld] &  &  \ar@{-}[uu] \ar@{.}[ll]       \\
 \ar@{-}[uu]\ar@{-}[rr] &  &  \ar@{-}[uu]\ar@{-}[ru]
} &
\xymatrix{
&  \ar@{-}[ld]\ar@{.}[dd] \ar@{-}[rr] & & \ar@{-}[ld]  \\
 \ar@{-}[rr] \ar@{-}[dd] &  &  \ar@{-}[dd] &            \\
& \alpha \ar@{.}[ld] &  & \beta \ar@{-}[uu] \ar@{.}[ll]       \\
 \ar@{-}[rr] &  &  \ar@{-}[ru]       \\
&  \ar@{.}[uu]\ar@{.}[ld] &  &  \ar@{-}[uu] \ar@{.}[ll]       \\
 \ar@{-}[uu]\ar@{-}[rr] &  &  \ar@{-}[uu]\ar@{-}[ru]
}
\end{tabular}
\caption{Two different kind of boundary constraints. $A$ and $B$ are the unknown symplectic pairs and the cubes represent the stabilizer generators that share support with the logical operator only at these two sites. \label{fig:boundary_constraint}}
\end{figure}
We shall simplify these expressions with the following notation.
\begin{defi}
\begin{align}
T_{\alpha \gamma} &=
\begin{pmatrix}
\alpha_1 & -\alpha_2 \\
\gamma_1 & -\gamma_2
\end{pmatrix} \\
T_{\alpha \gamma}^{\beta \delta} &=
\begin{pmatrix}
\beta_1 & -\beta_2 \\
\delta_1 & -\delta_2
\end{pmatrix}^{-1}
\begin{pmatrix}
\alpha_1 & -\alpha_2 \\
\gamma_1 & -\gamma_2
\end{pmatrix}
\end{align}
\end{defi}
The transition rule from $A$ to $B$ can be expressed as $T_{\beta \delta}^{-1} T_{\alpha \gamma}$. Note that the matrix is invertible since otherwise $\langle \alpha,  \gamma \rangle= 0$. Since the matrix acts on a vector space over finite field and it is invertible, it has a finite period. This is unlike in Haah's model which has a nilpotent transition matrix. From now on, we shall have a shortened notation for the transition rule.

Following are couple of useful facts about the transition matrix. We leave the proof for the readers.
\begin{align}
T^{AB}_{CD} &= T^{BA}_{DC} \\
(T^{AB}_{CD})^{-1} &= T^{CD}_{AB} \\
T^{AB}_{CD} T^{CD}_{EF} &= T^{AB}_{EF}
\end{align}
In this language, let us first find the condition for which a string logical operator with width $w=1$ is not allowed. There are two transition rules that may or may not be consistent to each other. Using the inversion symmetry, one can easily show that the transiton rules are $T$ and $T^{-1}$, where $T=T^{\gamma \beta}_{\delta \alpha}, T^{\delta \gamma}_{\alpha \beta}, T_{\alpha \delta}^{\gamma \beta}$ depending on the direction. In other words, given two consecutive symplectic pairs $\alpha$ and $\beta$, they must satisfy following two relations.
\begin{align}
\alpha &= T\beta \\
\alpha &= T^{-1} \beta.
\end{align}
This gives us
\begin{equation}
(T-T^{-1}) \beta = 0.
\end{equation}
If $T-T^{-1}$ is invertible, this implies $\beta=0$ and hence the string logical operator must be trivial.
\begin{lem}
If $\det(T-T^{-1})\neq 0$, where $T=T^{\gamma \beta}_{\delta \alpha}, T^{\delta \gamma}_{\alpha \beta}, T_{\alpha \delta}^{\gamma \beta}$, there is no string logical operator with width $w=1$.
\end{lem}

Surprisingly, adding a simple extra condition ensures the absence of string logical operator for arbitrarily large thickness. We defer the proof to Appendix \ref{appendix:proof}.
\begin{thm}
For $\mathcal{C}^{\alpha\beta \gamma \delta}_{S,A}$, maximum length of nontrivial string logical operator with width $w$ is bounded by $2w$, if the following conditions are satisfied.
\begin{itemize}
\item Deformability condition : $\langle A,B \rangle \neq 0$ $\forall A\neq B$, $A,B \in \{\alpha, \beta, \gamma, \delta \}$
\item Absence of minimal string : $\det(T-T^{-1}) \neq 0$ for $T=T^{\gamma \beta}_{\delta \alpha}, T^{\delta \gamma}_{\alpha \beta}, T_{\alpha \delta}^{\gamma \beta}$.
\item $\langle A,B\rangle^2 \neq \langle C,D\rangle^2$ $\forall A,B,C,D \in\{\alpha, \beta, \gamma, \delta \}$. $A,B,C,D$ are distinct.
\end{itemize}
\end{thm}

Within this framework, absence of string logical operator cannot be shown for $d=2,3$. For qubit case, this is exactly what we expect: one cannot even come up with a quantum code satisfying the first condition. The situation is subtly different for $d=3$. There are codes that satisfy first and second condition, but they do not satisfy the third condition. We have numerically checked the maximum length of nontrivial string logical operator for these codes. Up to width $w=20$, length is bounded by $w+1$. Hence we suspect these codes are free of string logical operator as well, but we do not have a machinery to prove it at this point.

It is important to note that the first and second condition is intimately tied together. We are studying a set of codes whose logical operator can be deformed into a flat segment. Within that framework, the second condition is necessary to ensure the absence of any string logical operator by definition. We conjecture the third condition can be reduced to a weaker condition.

\section{Properties of the code\label{sec:equivalence}}
Immediate consequence of Theorem 1 is that there is a large family of codes without string logical operator.  This becomes especially clear when $d$ is a large number. Which of them are equivalent to each other and which of them are not? What kind of logical operators do they have? We do not have a complete answer, but we will sketch the general properties which are shared by these codes.
\subsection{Equivalence relation}
Any two codes are equivalent to each other if they can be mapped via lattice symmetry and local unitary transformation. The lattice symmetry can be concisely represented as a permutation of the symplectic pairs $\alpha, \beta, \gamma, \delta$.
\begin{lem}
\begin{equation}
\mathcal{C}_{A,S}^{\alpha \beta \gamma \delta} =  \mathcal{C}_{A,S}^{\sigma(\alpha) \sigma(\beta) \sigma(\gamma) \sigma(\delta)},
\end{equation}
where $\sigma \in S_4$ over a set $\{\alpha, \beta, \gamma, \delta \}$.
\end{lem}

Local Clifford transformation can be represented as an element of $SL(2,p)$. One should also note that multiplying a nonzero element $a$ of $\text{GF}[p]$ does not change the code. It corresponds to merely changing the stabilizer element $s$ into $s^a$.
\begin{lem}
If $\exists a \in \mathbb{F}_p, M \in SL(2,p)$ such that $aM\{\alpha, \beta, \gamma, \delta \} = \{\alpha, \beta, \gamma, \delta \}$
\begin{equation}
\mathcal{C}_{A,S}^{\alpha \beta \gamma \delta} \cong \mathcal{C}_{A,S}^{\alpha' \beta' \gamma' \delta'},
\end{equation}
\end{lem}

Finally, there is a subtle equivalence relation between the antisymmetric and symmetric code. Instead of performing the same local unitary operation on all the qudits, one can imagine performing a unitary transformation on the even(or equivalently, odd) layer only, mapping $A \to -A$ for $A\in \{\alpha,\beta, \gamma,\delta \}$. This maps the symmetric code into the antisymmetric code and vice versa \emph{in the bulk.} However, if the length in the direction normal to these layers is odd, such operation illdefined. In other words, there exists a unitary operation that relates the antisymmetric and symmetric code in the bulk and even length scale.
\begin{lem}
\begin{equation}
\mathcal{C}_{S}^{\alpha \beta \gamma \delta} \cong \mathcal{C}_{A}^{\alpha \beta \gamma \delta}
\end{equation}
if the length of the system in $x$ direction is even, or in the bulk.
\end{lem}

Combining these equivalence relations together, we get following results.
\begin{corollary}
For $d=3$ code satisfying the deformability condition, there are two symmetric code and two antisymmetric code up to lattice symmetry and local Clifford operation. The parameters of the codes are $\{(1,0),(0,1),(1,1),(1,-1) \}$ and $\{(1,0),(0,1),(1,1),(-1,1) \}$.
\end{corollary}
\begin{corollary}
For $d=5$, $\{(1,0),(0,1),(1,1),(3,-3) \}$ is free of string logical operator with aspect ratio $5$.
\end{corollary}

\subsection{Logical Operators\label{sec:logicaloperator}}
Logical operators of the code is either a fractal or noncontractible surface. We shall study the surface operators in this paper. Depending on the system size, there are at least $1,2$ or $4$ surface operators for each directions, depending on the system size. Fortunately the pattern is periodic. Given a surface normal to one of the unit vectors $\hat{x},\hat{y},\hat{z}$, number of distinct surface operators normal to the vector depends on the width and height of the surface. If both of them are even, there are $4$ surface operators. If one of them is even, there are $2$. If none of them are even, there is $1$. The construction is quite straightforward.
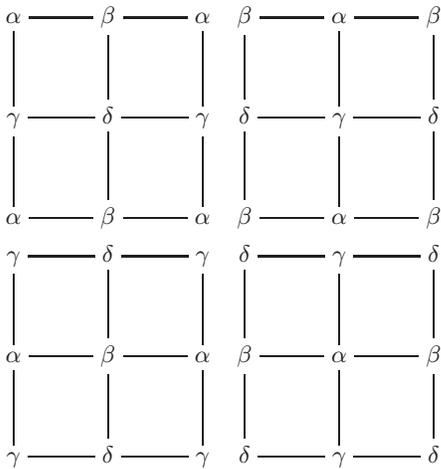
\begin{figure}[h]
\centering
\begin{tabular}{cc}
\xymatrix{
\alpha \ar@{-}[d] \ar@{-}[r] & \beta \ar@{-}[d] \ar@{-}[r] & \alpha \ar@{-}[d]   \\
\gamma \ar@{-}[d]\ar@{-}[r] & \delta \ar@{-}[d]\ar@{-}[r]& \gamma \ar@{-}[d]\\
\alpha  \ar@{-}[r] & \beta  \ar@{-}[r] & \alpha } &
\xymatrix{
\beta \ar@{-}[d] \ar@{-}[r] & \alpha \ar@{-}[d] \ar@{-}[r] & \beta \ar@{-}[d]   \\
\delta \ar@{-}[d]\ar@{-}[r] & \gamma \ar@{-}[d]\ar@{-}[r]& \delta \ar@{-}[d]\\
\beta  \ar@{-}[r] & \alpha  \ar@{-}[r] & \beta }
\\
\xymatrix{
\gamma \ar@{-}[d] \ar@{-}[r] & \delta \ar@{-}[d] \ar@{-}[r] & \gamma \ar@{-}[d]   \\
\alpha \ar@{-}[d]\ar@{-}[r] & \beta \ar@{-}[d]\ar@{-}[r]& \alpha \ar@{-}[d]\\
\gamma  \ar@{-}[r] & \delta  \ar@{-}[r] & \gamma } &
\xymatrix{
\delta \ar@{-}[d] \ar@{-}[r] & \gamma \ar@{-}[d] \ar@{-}[r] & \delta \ar@{-}[d]   \\
\beta \ar@{-}[d]\ar@{-}[r] & \alpha \ar@{-}[d]\ar@{-}[r]& \beta \ar@{-}[d]\\
\delta  \ar@{-}[r] & \gamma  \ar@{-}[r] & \delta }
 \end{tabular}
 \caption{Construction of logical operators on a plane. Each of them can be mapped into each other via unit translation.\label{fig:LogicalOperator_4}}
\end{figure}

If the width and height of the surface are both even, all the periodic structures shown in FIG.\ref{fig:LogicalOperator_4} are allowed. Otherwise, none of them are allowed. It is still possible, however, to construct logical operators by multiplying two of the logical operators in FIG.\ref{fig:LogicalOperator_4}. For instance, multiplying the first and second or third and fourth in FIG.\ref{fig:LogicalOperator_4} results in a periodic structure in $x$ direction. Similarly, multiplying first and third or second and fourth results in a periodic strucrture in $y$ direction. Similarly, when both the width and height are odd number, multiplying all $4$ of them results in a periodic structure in both $x$ and $y$ direction. One can apply the same logic for all three different directions.

One final note we would like to point out is that the antisymmetric codes always ensure at least one encoded qudit. Given a system with $n$ qudits, there are $n$ cubic stabilizer generators. There is at least one nontrivial relation between the generators: multiplication of all the generators equals identity. The same logic cannot be applied to the symmetric code, but at certain system sizes, it is possible to show the existenece of nontrivial logical qudit by computing the commutation relations between the planar logical operators.

\section{Conclusion\label{sec:conclusion}}
We presented a family of nonbinary quantum codes without string logical operator. This tells us that systems with such properties exist in a more general context, but interesting questions still remain. How do the defects interact with each other? Is there a generalization to nonabelian group as in quantum double model? These questions are hard to answer, for the movement of the defects are constrained. Conventional argument for Boson-Fermion dichotomy in 3D does not hold: simple exchange of two defects cannot be performed in general without creating extra defects. It would be interesting to see if one can come up with a similar nontrivial `braiding rule.' Such rule would likely to incorporate a collection of particles rather than an exchange of two particles.
\begin{acknowledgements}
This research was supported in part by NSF under Grant No. PHY-0803371, by ARO Grant No. W911NF-09-1-0442, and DOE Grant No. DE-FG03-92-ER40701. Author would like to thank Steven Flammia, Jeongwan Haah, Robert Koenig, and John Preskill for useful discussions.
\end{acknowledgements}
\bibliography{bib}

\begin{thebibliography}{13}
\expandafter\ifx\csname natexlab\endcsname\relax\def\natexlab#1{#1}\fi
\expandafter\ifx\csname bibnamefont\endcsname\relax
  \def\bibnamefont#1{#1}\fi
\expandafter\ifx\csname bibfnamefont\endcsname\relax
  \def\bibfnamefont#1{#1}\fi
\expandafter\ifx\csname citenamefont\endcsname\relax
  \def\citenamefont#1{#1}\fi
\expandafter\ifx\csname url\endcsname\relax
  \def\url#1{\texttt{#1}}\fi
\expandafter\ifx\csname urlprefix\endcsname\relax\def\urlprefix{URL }\fi
\providecommand{\bibinfo}[2]{#2}
\providecommand{\eprint}[2][]{\url{#2}}

\bibitem[{\citenamefont{Kitaev}(2003)}]{Kitaev2003}
\bibinfo{author}{\bibfnamefont{A.~Y.} \bibnamefont{Kitaev}},
  \bibinfo{journal}{Annals Phys.} \textbf{\bibinfo{volume}{303}},
  \bibinfo{pages}{2} (\bibinfo{year}{2003}), \eprint{quant-ph/9707021}.

\bibitem[{\citenamefont{Bombin and Martin-Delgado}(2006)}]{Bombin2006}
\bibinfo{author}{\bibfnamefont{H.}~\bibnamefont{Bombin}} \bibnamefont{and}
  \bibinfo{author}{\bibfnamefont{M.~A.} \bibnamefont{Martin-Delgado}},
  \bibinfo{journal}{Phys.Rev.Lett.} \textbf{\bibinfo{volume}{97}},
  \bibinfo{pages}{180501} (\bibinfo{year}{2006}).

\bibitem[{\citenamefont{Bombin and Martin-Delgado}(2007)}]{Bombin2007}
\bibinfo{author}{\bibfnamefont{H.}~\bibnamefont{Bombin}} \bibnamefont{and}
  \bibinfo{author}{\bibfnamefont{M.~A.} \bibnamefont{Martin-Delgado}},
  \bibinfo{journal}{Phys.Rev.B} \textbf{\bibinfo{volume}{75}},
  \bibinfo{pages}{075103} (\bibinfo{year}{2007}).

\bibitem[{\citenamefont{Yoshida}(2011)}]{Yoshida2011}
\bibinfo{author}{\bibfnamefont{B.}~\bibnamefont{Yoshida}}
  (\bibinfo{year}{2011}), \eprint{1103.1885}.

\bibitem[{\citenamefont{Nussinov and Ortiz}(2008)}]{Nussinov2008}
\bibinfo{author}{\bibfnamefont{Z.}~\bibnamefont{Nussinov}} \bibnamefont{and}
  \bibinfo{author}{\bibfnamefont{G.}~\bibnamefont{Ortiz}},
  \bibinfo{journal}{Phys. Rev. B} \textbf{\bibinfo{volume}{77}},
  \bibinfo{pages}{064302} (\bibinfo{year}{2008}).

\bibitem[{\citenamefont{Nussinov and Ortiz}(2009)}]{Nussinov2009}
\bibinfo{author}{\bibfnamefont{Z.}~\bibnamefont{Nussinov}} \bibnamefont{and}
  \bibinfo{author}{\bibfnamefont{G.}~\bibnamefont{Ortiz}},
  \bibinfo{journal}{Annals of Physics} \textbf{\bibinfo{volume}{324}},
  \bibinfo{pages}{977 } (\bibinfo{year}{2009}).

\bibitem[{\citenamefont{Kim}(2011)}]{Kim2011}
\bibinfo{author}{\bibfnamefont{I.~H.} \bibnamefont{Kim}},
  \bibinfo{journal}{Phys. Rev. A} \textbf{\bibinfo{volume}{83}},
  \bibinfo{pages}{052308} (\bibinfo{year}{2011}).

\bibitem[{\citenamefont{Alicki et~al.}(2008)\citenamefont{Alicki, Horodecki,
  Horodecki, and Horodecki}}]{Alicki2008}
\bibinfo{author}{\bibfnamefont{R.}~\bibnamefont{Alicki}},
  \bibinfo{author}{\bibfnamefont{M.}~\bibnamefont{Horodecki}},
  \bibinfo{author}{\bibfnamefont{P.}~\bibnamefont{Horodecki}},
  \bibnamefont{and}
  \bibinfo{author}{\bibfnamefont{R.}~\bibnamefont{Horodecki}},
  \bibinfo{journal}{Open Syst. Inf. Dyn.} \textbf{\bibinfo{volume}{17}},
  \bibinfo{pages}{1} (\bibinfo{year}{2008}), \eprint{0811.0033}.

\bibitem[{\citenamefont{Haah}(2011)}]{Haah2011}
\bibinfo{author}{\bibfnamefont{J.}~\bibnamefont{Haah}}, \bibinfo{journal}{Phys.
  Rev. A} \textbf{\bibinfo{volume}{83,}}, \bibinfo{pages}{042330}
  (\bibinfo{year}{2011}), \eprint{1101.1962}.

\bibitem[{\citenamefont{Bravyi and Haah}(2011{\natexlab{a}})}]{Bravyi2011}
\bibinfo{author}{\bibfnamefont{S.}~\bibnamefont{Bravyi}} \bibnamefont{and}
  \bibinfo{author}{\bibfnamefont{J.}~\bibnamefont{Haah}},
  \bibinfo{journal}{Phys. Rev. Lett.} \textbf{\bibinfo{volume}{107,}},
  \bibinfo{pages}{150504} (\bibinfo{year}{2011}{\natexlab{a}}),
  \eprint{1105.4159}.

\bibitem[{\citenamefont{Bravyi and Haah}(2011{\natexlab{b}})}]{Bravyi2011a}
\bibinfo{author}{\bibfnamefont{S.}~\bibnamefont{Bravyi}} \bibnamefont{and}
  \bibinfo{author}{\bibfnamefont{J.}~\bibnamefont{Haah}}
  (\bibinfo{year}{2011}{\natexlab{b}}), \eprint{1112.3252}.

\bibitem[{\citenamefont{Ketkar et~al.}(2006)\citenamefont{Ketkar, Klappenecker,
  Kumar, and Sarvepalli}}]{Ketkar2006}
\bibinfo{author}{\bibfnamefont{A.}~\bibnamefont{Ketkar}},
  \bibinfo{author}{\bibfnamefont{A.}~\bibnamefont{Klappenecker}},
  \bibinfo{author}{\bibfnamefont{S.}~\bibnamefont{Kumar}}, \bibnamefont{and}
  \bibinfo{author}{\bibfnamefont{P.}~\bibnamefont{Sarvepalli}},
  \bibinfo{journal}{Information Theory, IEEE Transactions on}
  \textbf{\bibinfo{volume}{52}}, \bibinfo{pages}{4892 } (\bibinfo{year}{2006}),
  ISSN \bibinfo{issn}{0018-9448}.

\bibitem[{\citenamefont{Pittenger and Rubin}(2004)}]{Pittenger2004}
\bibinfo{author}{\bibfnamefont{A.~O.} \bibnamefont{Pittenger}}
  \bibnamefont{and} \bibinfo{author}{\bibfnamefont{M.~H.} \bibnamefont{Rubin}},
  \bibinfo{journal}{Linear Alg. Appl.} \textbf{\bibinfo{volume}{390,}},
  \bibinfo{pages}{255} (\bibinfo{year}{2004}), \eprint{quant-ph/0308142}.

\end{thebibliography}
\appendix
\section{Terminologies from Haah's paper}
There are recurring terminologies and ideas used throughout the recent result of Bravyi and Haah. We shall present a brief overview of these materials which are particularly relevant to our work. One of the defining properties of Haah's code is the absence of string logical operator. String logical operator can be broken into \emph{logical string segment.} Following set of definitions are from Haah's original paper.\cite{Haah2011}
\begin{defi}(Haah 2011)
A set of sites $\{p_1,p_2, \cdots, p_n \}$ is a \emph{path} joining $p_1$ and $p_n$ if for each pair $(p_i,p_{i+1})$ of consecutive sites there exists a stabilizer generators that acts nontrivially on their pair simultaneously, for $i=1, \cdots, n-1$. A set $M$ of sites is \emph{connected} if every pair of sites in $M$ are joined by a path in $M$. A \emph{connected Pauli operator} is a Pauli operator with connected support.
\end{defi}
\begin{defi}(Haah 2011)
Let $\Omega_1, \Omega_2$ be congruent cubes consisting of $w^3$ sites, and $O$ be a finite Pauli operator. A triple $\eta =(O,\Omega_1, \Omega_2)$ is a logical string segment if every stabilizer generator that acts trivially on both $\Omega_1$ and $\Omega_2$ commutes with $O$. We call $\Omega_{1,2}$ the \emph{anchor}. The \emph{directional vector} of $\eta$ is the relative position of $\Omega_1$ to $\Omega_2$. The \emph{length} is the $l_1$-length of the directional vector, and the \emph{width} is $w$.
\end{defi}
\begin{defi}(Haah 2011)
A logical string segment $\eta=(O, \Omega_1, \Omega_2)$ is \emph{connected} if there exists two sites $p_1 \in \Omega_1$, $p_2, \Omega_2$ that can be joined by a path in $supp(O) \cup \{p_1,p_2 \}$, where $supp(O)$ is a set of sites on which $O$ acts nontrivially. Two logical string segments $(O, \Omega_1, \Omega_2)$, $(O',\Omega_1, \Omega_2)$ are equivalent if $O'$ can be obtained from $O$ by multiplying finitely many stabilizer generators. $\eta$ is nontrivial if every equivalent logical string segment is connected.
\end{defi}

\section{Proof of Theorem 1\label{appendix:proof}}
We shall prove Theorem 1. Let us start with basics of linear algebra.
Let $V$ be a $n$-dimensional vector space over a finite field $\mathbb{F}$. $T$ is a linear operator $T:V\to V$ and $v\in V$.
\begin{defi}
$m_{T,v}(x)$ is a polynomial with least degree which satisfies the relation
\begin{equation}
m_{T,v}(T)v=0.
\end{equation}
We call $m_{T,v}$ as a minimal polynomial of $T$ on $v$.
\end{defi}
\begin{defi}
$m_T(x)$ is a polynomial with least degree which satisfies the relation
\begin{equation}
m_T(T) v=0 \quad \forall v \in V
\end{equation}
\end{defi}
\begin{lem}
\begin{align}
m_{T,v} &| m_T \\
m_T &| \chi_T,
\end{align}
where $\chi_T$ is a characteristic polynomial of $T$.
\end{lem}
\begin{lem}
Let the degree of $m_{T,v}$ be $d$. Then $v, Tv, T^2v, \cdots, T^{d-1}v$ are linearly independent to each other.
\end{lem}

The sketch of the proof is the following. Suppose we are given a logical string segment with bounded anchor size. The segment must commute with the local stabilizer generators. As the length of the segment increases, the number of such constraints increases, since more local stabilizer generators must commute. At the same time, number of unknowns to specify the segment increases as well. One can show that the rate of increase for the constraints is larger than that of the number of unknowns. Hence eventually the number of constraints overcome the number of unknowns. If the constraints are sufficiently independent to each other, rank of the constraints can become identical to the number of unknowns. In such cases, logical string segment satisfying the commutation relation must be trivial.

The most technically nontrivial part is proving the independence of the constraints. We shall explain all the steps in full detail.

\subsection{Enumeration of local constraints : Flat segment}
We shall first start with a simpler case. Suppose logical string segment can be supported on a flat surface which is normal to one of three $\hat{x},\hat{y},\hat{z}$ orthogonal directions. Logical string segment can be represented as a set of symplectic pairs on the vertices. Given a width $w$ and length $l$ segment, we have total of $wl$ symplectic pairs, which results in $2wl$ unknowns over a field $\mathbb{F}$. We shall represent the logical operator as a $2wl$-dimensional vector over $\mathbb{F}$.
\begin{defi}
$i$th row, $j$th column, $k$th element of the symplectic pair is labeled by $2(j-1)n+2(i-1)+k$.
\end{defi}
See FIG.\ref{fig:basisvector} for example.
\begin{figure}[h]
\centering
\begin{tabular}{c}
\xymatrix{
(1,2) \ar@{-}[d] \ar@{-}[r] & (7,8) \ar@{-}[d] \ar@{-}[r] & (13,14) \ar@{-}[d] \ar@{-}[r]&  (19,20) \ar@{-}[d]   \\
(3,4) \ar@{-}[d]\ar@{-}[r] & (9,10) \ar@{-}[d]\ar@{-}[r]& (15,16) \ar@{-}[d] \ar@{-}[r]&  (21,22) \ar@{-}[d]\\
(5,6)  \ar@{-}[r] & (11,12)  \ar@{-}[r] & (17,18)  \ar@{-}[r] & (23,24) }
 \end{tabular}
 \caption{Notation for the basis vectors\label{fig:basisvector}}
\end{figure}
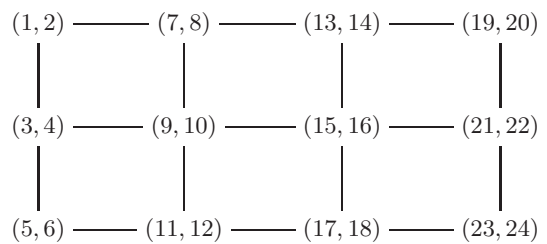

Within this convention, constraints from the local stabilizer generators takes the form
\begin{equation}
c^T v_{L}=0,
\end{equation}
where $v_{L}$ is the vector representing the logical operator and $c$ is the constraint. For instance, consider $w=1$, $l=2$ string segment. Set of local constraints can be represented by the following matrix, which we shall call as \emph{constraint matrix}.
\begin{equation}
\begin{pmatrix}
T_{\gamma \beta}& T_{\delta \alpha}  \\
T_{\delta \alpha} & T_{\gamma \beta}
\end{pmatrix},
\end{equation}
where each entries are $2\times 2$ blocks. The rank of this matrix is preserved under block Gaussian elimination. After some manipulation, the matrix can be transformed into a block upper triangular form.
\begin{equation}
\begin{pmatrix}
I & T^{\gamma \beta}_{\delta \alpha} \\
0 & T^{\delta \alpha}_{\gamma \beta} -T^{\gamma \beta}_{\delta \alpha}
\end{pmatrix}.
\end{equation}
This matrix is full rank if and only if $T^{\delta \alpha}_{\gamma \beta} -T^{\gamma \beta}_{\delta \alpha}$ is rank-2. Hence we arrive at the following conclusion.
\begin{lem}
String logical operator with width $1$ exists if and only if
\begin{equation}
\det (T^{\delta \alpha}_{\gamma \beta} -T^{\gamma \beta}_{\delta \alpha}) \neq 0.
\end{equation}
\end{lem}

Pattern emerges as we increase the width of the segment. For instance, constraint matrix for $w=2$, $l=3$ is the following.
\begin{equation}
\begin{pmatrix}
T_{\gamma \beta}& 0 &  T_{\delta \alpha} & 0 & 0 & 0\\
T_{\alpha \delta} & T_{ \gamma \beta} & T_{ \beta \gamma} & T_{\delta \alpha} & 0 & 0\\
0 & T_{\delta \alpha} & 0 & T_{\gamma \beta} & 0 & 0\\
0 & 0 & T_{\gamma \beta}& 0 &  T_{\delta \alpha} & 0 \\
0 & 0 & T_{ \alpha \delta} & T_{\gamma \beta } & T_{ \beta \gamma} & T_{ \delta \alpha} \\
0 & 0 & 0 &  T_{\delta \alpha} & 0 & T_{\gamma \beta}
\end{pmatrix}\label{eq:canonical_form}
\end{equation}
After a sequence of suitable block Gaussian elimination, we can arrive at the following canonical form. Here the width $w$ was set to $n$.
\begin{equation}
\begin{pmatrix}
I_{2n} & \mathcal{T}_{2n} & 0 & 0 & \cdots & 0&0 \\
0 & I_{2n} & \mathcal{T}_{2n} & 0 & \cdots & 0&0 \\
0 & 0 & I_{2n} & \mathcal{T}_{2n} & \cdots & 0&0 \\
\vdots & \vdots &\vdots &\vdots &\ddots & 0&0  \\
0 & 0 & 0 & 0 & \cdots & I_{2n} & \mathcal{T}_{2n} \\
0 &v_{2n} &0 & 0 & \cdots &0 &0\\
0 &0 &v_{2n} & 0 & \cdots &0 &0\\
\vdots & \vdots &\vdots &\vdots &\ddots & 0&0  \\
0 &0 &0 & 0 & \cdots &0 &v_{2n}
\end{pmatrix},
\end{equation}
where
\begin{equation}
\mathcal{T}_{2n} =
\begin{pmatrix}
T^{\gamma \beta}_{\delta \alpha} & 0 & \cdots & 0 \\
(-T^{\gamma \beta}_{\alpha \delta})x & T^{\gamma \beta}_{\delta \alpha} & \cdots & 0\\
\vdots & \vdots & \ddots & 0 \\
(-T^{\gamma \beta}_{\alpha \delta})^{n-1}x & (-T^{\gamma \beta}_{\alpha \delta})^{n-2}x & \cdots & T^{\gamma \beta}_{\delta \alpha}
\end{pmatrix}
\end{equation}
and
\begin{equation}
v_{2n} =
\begin{pmatrix}
(T^{\gamma \beta}_{\alpha\delta})^{n-1}x & (T^{\gamma \beta}_{\alpha\delta})^{n-2}x & \cdots & (T^{\gamma \beta}_{\alpha\delta})^{1}x & x
\end{pmatrix}.
\end{equation}
$I_{2n}$ is a $2n \times 2n$ identity matrix. Given a length $l$, the dimension of the constraint matrix is $2(w+1)(l-1) \times 2wl$. The rank of the constraint matrix can be bounded by $2w(l-1)+ Rank(\mathcal{A}_{2n})$, where
\begin{equation}
\mathcal{A}_{2n}=\begin{pmatrix}
v_{2n}  \\
v_{2n} \mathcal{T}_{2n} \\
\vdots \\
v_{2n} (\mathcal{T}_{2n})^{2n-2} \\
v_{2n} (\mathcal{T}_{2n})^{2n-1}
\end{pmatrix}.
\end{equation}

\begin{lem}
Let $v_{2n}^{1,2}$ be the first and second column vector of $v_{2n}$.
$\max_{1,2}(\deg(m_{\mathcal{T}_{2n},v_{2n}^{1,2}}))=2n$ if the conditions in Theorem 1 is met.
\end{lem}
\begin{proof}
\begin{equation}
\chi_{\mathcal{T}_{2n}}= \chi_{T^{\gamma \beta}_{\delta \alpha}}^n.
\end{equation}
Under the conditions in Theorem 1,
\begin{equation}
\chi_{T^{\gamma \beta}_{\delta \alpha}}(\mathcal{T}_{2n}) =
\begin{pmatrix}
0 & 0 & \cdots & 0 & 0 \\
A & 0 & \cdots & 0 & 0\\
B & A & \cdots & 0 & 0 \\
\vdots & \vdots & \ddots & \vdots & \vdots \\
D & C & \cdots & A & 0
\end{pmatrix},
\end{equation}
where $\det(A) \neq 0$. Note that this step is not trivial. 

One can then show that
\begin{equation}
v_{2n}\chi_{T^{\gamma \beta}_{\delta \alpha}}(\mathcal{T}_{2n})^{n-1} \neq 0,
\end{equation}
since the first $2\times 2$ block of the matrix is of form $(T^{\gamma \beta}_{\alpha\delta})^{n-1}x A^{n-1}$, and this is a product of invertible matrices. Hence the minimal polynomial of $\mathcal{T}_{2n}$ cannot have a degree of $2(n-1)$. If $\chi_{T^{\gamma \beta}_{\delta \alpha}}(x)$ is irreducible, we are done. Otherwise, consider a nontrivial factor of $\chi_{T^{\gamma \beta}_{\delta \alpha}}(x)$ and denote it as $f(x)=x+a$. Consider the first $2 \times 2$ block of $v_{2n}\chi_{T^{\gamma \beta}_{\delta \alpha}}(\mathcal{T}_{2n})^{n-1} (\mathcal{T}_{2n}+a)$. It has the form of $x A^{n-1}A'$, where $A'$ is a rank one matrix. Hence at least one of $v_{2n}^{1,2}$ has a minimal polynomial for $\mathcal{T}_{2n}$ with degree $2n$. This proof holds only when the logical string segment is supported on a plane. A generalized proof will be presented in the next section.
\end{proof}

\subsection{Enumeration of local constraints : Flat segment with a corner}
The proof idea is identical to that for the flat segment. Given a set of constraints, one can bound the linear independence of these constraints by determining the minimal polynomial of certain matrix. Under the same procedure, Eq.\ref{eq:canonical_form} can be derived, but $\mathcal{T}_{2n}$ and $v_{2n}$ are changed. Their precies forms are not so concise, but for the proof only the following information is necessary. First, the first $2 \times 2$ block of $v_{2n}$ is $(T^{\gamma \alpha}_{\beta \delta})^{w-w_1}T_{int}(T^{\gamma \beta}_{\alpha \delta})^{w_1-1}x$, where $T_{int}=T^{\gamma \alpha}_{\delta 0} - T^{\gamma \alpha}_{\beta \delta}T^{\gamma \alpha}_{\alpha 0}$. Important property of $T_{int}$ is that $\det (T_{int})=\langle \alpha, \gamma \rangle \langle \alpha, \delta \rangle \langle \delta, \alpha \rangle \neq 0$ due to the deformability condition. Second, $\mathcal{T}_{2n}$ is a $2\times 2$-block lower triangular form with following entries.
\begin{align}
(\mathcal{T}_{2n})_{ii} &= T^{\gamma \beta}_{\delta \alpha} \qquad i< w_1\\
&= T^{\gamma \alpha}_{\delta \beta} \qquad i \geq w_1,
\end{align}
\begin{align}
(\mathcal{T}_{2n})_{(i+1)i} &= -T^{\gamma \beta}_{\alpha\delta} x \qquad i< w_1\\
&= -T^{\gamma \alpha}_{\beta\delta }x \qquad i > w_1,
\end{align}
where $w_1$ is the width before we encounter the corner. Rest of the entries can be computed as well, but they are irrelevant for the proof. Argument goes as follows. $\chi_{\mathcal{T}_{2n}}(x)= \chi_{T^{\gamma \beta}_{\delta \alpha}}(x)^{w_1-1} \chi_{T^{\gamma \alpha}_{\delta \beta}}(x)^{w-w_1+1}$. The minimal polynomial must contain the factor of $\chi_{T^{\gamma \beta}_{\delta \alpha}}(x)^{w_1-1}$. Otherwise, the first $2\times 2$ block of $v_{2n}\chi_{^{\gamma \beta}_{\delta \alpha}}(\mathcal{T}_{2n})$ is an invertible matrix. Consider a polynomial $g(x)=\chi_{^{\gamma \beta}_{\delta \alpha}}(x)^{w_1-1} \chi_{T^{\gamma \alpha}_{\delta \beta}}(x)^{w-w_1}$. The $w_1+1$th $2\times 2$ block of $v_{2n}g(\mathcal{T}_{2n})$ is nonzero, for it has the form of $(T^{\gamma \alpha}_{\beta\delta})^{n-1}x A^{w-w_1}$ for some invertible matrix $A$. If $\chi_{T^{\gamma \alpha}_{\delta \beta}}(x)$ is irreducible, we are done. Otherwise, use the same argument as in Lemma 9.
\end{document}